\documentclass[pdflatex,iicol,sn-mathphys-num]{sn-jnl}

\usepackage[T1]{fontenc}
\usepackage{graphicx}
\usepackage{booktabs}
\usepackage{amsmath,amssymb}
\usepackage{array}
\usepackage{url}
\usepackage{microtype}
\usepackage{tikz}
\usetikzlibrary{arrows.meta,fit}
\microtypesetup{expansion=false}
\hypersetup{pdfauthor={Qunhui Zhang},pdftitle={AiFlow: Construction-to-Retirement Conformance for Versioned LLM Services}}

\begin{document}

\title[AiFlow construction and retirement contracts]{AiFlow: Construction-to-Retirement Conformance for Versioned LLM Services}

\author*{\fnm{Qunhui} \sur{Zhang}}\email{will\_zhang@sjtu.edu.cn}
\affil{\orgdiv{School of Software}, \orgname{Shanghai Jiao Tong University}, \orgaddress{\city{Shanghai}, \country{China}}}

\abstract{A live update can preserve the topology and output of a large language model (LLM) workflow while bypassing a filter's queue and concurrency limits. A background task can use the old version after its caller returns. AiFlow links construction, invocation and retirement through an executable contract. Each admitted request captures an action bound to the checked operators. The service certifies the call's terminal outcome and the release of its registered work. Once new calls can no longer capture the version, the owner checks these certificates and remaining local resources before retirement. Strict publication pauses admission while old calls drain. Deferred publication can admit calls to the successor before its predecessor retires. We evaluated two stateless linear workflows in one Java virtual machine. Under strict publication, we repeated three injected fault scenarios three times per check configuration. An observation counted as contained if neither unguarded execution nor premature replacement occurred. With construction or completion checking alone, six of nine fault observations met this criterion. With both checks, all nine did. A separately implemented coordinator matched AiFlow on the lifecycle outcome, service call count and reclamation decision in all 24 tested pairs. With updates every 400 ms, deferred publication met the 500 ms deadline more often at both loads. We measured the share of all offered requests that met the deadline on the latest scheduled version due at return. Strict publication led at 80 requests/s and deferred publication at 20 requests/s. Both policies use the same local evidence for retirement.}

\keywords{service orchestration, LLM applications, workflow conformance, dynamic update, admission control, middleware}

\maketitle

\section{Introduction}\label{sec:intro}

Services that compose model generation with retrieval or postprocessing \cite{ref54,ref25} may replace a filter while requests are in flight. A call that reaches the filter must execute the operator checked for its version. The version must remain available while local work still uses it.

The app-platform change \texttt{f244ef0} adds reasoning-block filtering before recommendation parsing \cite{refAppPlatform}. Consider two FIT Expression for LLM (FEL) workflows. In $F$, a Node Guardian limits the filter's queue and worker concurrency. In $F'$, an ordinary map returns the same JavaScript Object Notation (JSON) output along the same named edges. A service that checks $F$ and invokes $F'$ can balance every lease and still bypass the guard. If it releases a lease when the caller returns, the count can also reach zero while an audit is using the version.

AiFlow's Guardian and compiler enforce local capacity policies \cite{refAiFlow}. Preserving those policies through replacement requires evidence that connects the checked operator to the last local task using its version.
We give an executable contract with three obligations. At construction, the factory binds the checked operators to the action a call can capture. At completion, the service certifies that call's terminal outcome and the release of its registered work. For retirement, the owner excludes new captures, requires a clean release for every lease and verifies the version's local resources. Strict and deferred publication differ in when they admit calls to a successor. They can use the same evidence to retire its predecessor.

We test these obligations by substituting a filter, holding an audit open and comparing lifecycle decisions with an independently implemented coordinator. We measure how each publication policy affects timely service and adoption of scheduled versions.

\section{Service setting and running workflow}\label{sec:setting}

One Java virtual machine (JVM) hosts named, stateless linear workflows. The adapter enforces queue and worker limits and rejects declarations for retained state, feedback and edge buffers. The Java Development Kit (JDK) reference covers unordered blocking admission.

In the recommendation workflow, native and fixed answers enter the local business path outlined in Fig.~\ref{fig:workflow}. The guarded \texttt{filterReasoningContent} strips complete reasoning blocks. A loopback Hypertext Transfer Protocol (HTTP) endpoint validates the array and returns canonical JSON. The workflow checks the expected array and tags delivery with the captured version and request ID.

Fixed inputs exercise reasoning-block removal. Structured model answers exercise passthrough. The native trial harness holds the outer lease across the provider exchange and the business call.

\begin{figure}[t]
\centering
\begin{tikzpicture}[font=\fontsize{8}{9.5}\selectfont,>=Stealth,
  stage/.style={draw=black!65,rounded corners=2pt,align=center,
    text width=2.48cm,minimum height=6mm,inner sep=2pt},
  branch/.style={stage,fill=black!5},
  data/.style={->,line width=.45pt},
  exchange/.style={<->,dashed,line width=.45pt}]
\node[branch] (native) at (-1.48,0) {Native trial\\Hold outer lease};
\node[branch] (fixed) at (1.48,0) {Fixed-answer input};
\node[stage] (reader) at (-1.48,-1.08) {Read terminal\\Close local body};
\node[stage,text width=1.42cm] (ollama) at (1.51,-1.08) {Ollama\\server};
\node[stage,text width=4.88cm] (enter) at (0,-2.17) {Capture business version lease};
\node[stage] (filter) at (-1.48,-3.19) {Guarded\\reasoning filter};
\node[stage] (http) at (1.48,-3.19) {Loopback HTTP\\JSON service};
\node[stage] (audit) at (1.48,-4.25) {Audit and\\tagged delivery};
\node[stage] (certify) at (-1.48,-4.25) {Certify FEL completion};
\node[draw=black!55,dashed,rounded corners=2pt,inner sep=4pt,
  fit=(enter)(filter)(http)(audit)(certify)] {};
\draw[data] (native) -- (reader);
\draw[exchange] (reader) -- node[above,inner sep=1pt] {HTTP} (ollama);
\draw[data] (reader.south) -- (reader.south |- enter.north);
\draw[data] (fixed.east) -- (3.03,0) |- (enter.east);
\draw[data] (filter.north |- enter.south) -- (filter.north);
\draw[data] (filter) -- (http);
\draw[data] (http) -- (audit);
\draw[data] (audit) -- (certify);
\end{tikzpicture}
\caption{Recommendation workflow}\label{fig:workflow}
\end{figure}

An ingress semaphore caps concurrent calls across versions. Both policies refuse entry when permits run out. An uncertain return releases the caller permit and closes admission. The owner checks outstanding work before disposal. Shared executors and services must outlive all versions that use them.

\section{Related work}\label{sec:related}

LangChain, LangGraph and LlamaIndex support composition and streaming \cite{ref23,ref24,ref26}. DSPy compiles model programs \cite{ref19}. LangGraph supports persistence and retry \cite{ref57,ref58}. Studies of service composition document the integration challenges developers face \cite{ref56,ref16}.

Runtime discovery reconstructs architecture from execution \cite{ref65}. SAGE compiles LLM pipelines with backpressure and resource contracts \cite{ref63}. Agentproof checks graph structure and temporal policies over traces \cite{refAgentproof}. In AiFlow, the identity of the installed FEL subscriber distinguishes the guarded filter from a substitute with the same edges and output.

Reactive Streams and Reactor supply demand and overflow mechanisms \cite{ref33,ref44,ref59}. The Dataflow Model and Structured Streaming address records, time and state \cite{ref2,ref5}. Fries orders operator reconfiguration so that in-flight tuples see transactionally consistent changes \cite{refFries}. Service-process reconfiguration also addresses end-to-end quality constraints \cite{ref55}.

K42 blocks invocations for quiescent replacement \cite{refK42}. Other approaches study update-point timing \cite{refQuiescence} and state transformation \cite{refKitsune}. Read-copy-update delays reclamation until prior readers finish \cite{refRCU}. In our workflow the audit can outlive the call that started it. The completion adapter supplies the reader-exit evidence for this case by tracking work owned by the invocation.

ADEPTflex checks structural changes to running instances \cite{refADEPT}. ConUp preserves component-version consistency across a distributed transaction \cite{refVersionConsistency}. Essential Safety uses the semantics of workflow changes to identify safe update points \cite{refEssentialSafety}. Temporal's Workflow Pinning keeps an execution on one Worker Deployment Version \cite{refTemporal}. Within one version, $F$ and $F'$ share named edges. AiFlow checks which subscriber the action invokes and tracks work that can outlive the call.

Exclusive builders establish operator identity, and JDK futures track task completion \cite{refJdkFuture,refJdkPool}. A future for $F'$ can finish after a callback checked $F$. Java structured concurrency also gives an invocation a way to account for child tasks. Closing a task scope waits for its subtasks to finish \cite{refTaskScope}. AiFlow binds the invocation to the checked operators of its captured version. Once that version is deselected, its owner uses each invocation's terminal and release evidence together with a final resource check to decide when to reclaim it. Section~\ref{sec:reference} tests these obligations with immutable handles and invocation-owned futures.

\section{Construction, completion, and retirement contract}\label{sec:design}

\subsection{Construction evidence and a closed activation boundary}

A version carries $C_v=(F_v,D_v,R_v)$: its actual FEL flow, typed/policy declaration, and construction records. Activation checks unique nonblank node IDs, the declared node set, direct subscriptions with matching multiplicities, Java class witnesses and policies. Each guarded map records the flow and subscriber created by \texttt{bind}; the checked candidate must contain that same flow/subscriber pair. Installation completes before the owner admits traffic.

Let $P(F)$ retain node IDs, direct-edge multiplicities and declared witnesses. A guarded flow and an ordinary-map rebuild can have equal $P$ although only one enters Guardian. A classifier restricted to $P$ cannot distinguish this pair. Construction records or exclusive builder ownership distinguish the executing operators even when their named topologies coincide.

The lower-level coordinator accepts separate preparation and action callbacks, so a caller can check $F$ and then execute $F'$ or change an exposed flow. In the closed interface, each factory builds and activates its graph and retains private, final references to its declaration, flow and guards. It returns a \texttt{CertifiedRecommendationVersion}. The public \texttt{Service.replace} operation accepts that version, whose action always invokes the graph the factory checked.

Each version has one owner. Responses carry a version label and sequence number (\texttt{label@sequence}). Refused candidates can leave gaps. Clients may discard unclaimed versions. Owners retain unresolved terminal handles for shutdown.

The document projection factory declares and guards a field selector. Activation rejects a queue mismatch. Projection calls use the same invocation context, completion adapter and retirement rule as recommendation calls (Eq.~\ref{eq:retire}).

\subsection{Invocation evidence}

Each guarded map binds to one stage, and activation rejects duplicate guard IDs. For invocation $i$ on version $v$, $G_v$ is the set of expected guard instances and $V_i$ the set that enrolled. Enrollment accepts only expected identities, so $V_i\subseteq G_v$ throughout execution. Each enrollment issues a separate ticket. Write $D_i$ when enrollment is sealed and all issued tickets have been released, including those for rejected or purged attempts.

Let $N_i$ denote a normal terminal return after delivery to the consumer. Let $M_i$ denote a failed wait whose direct cause is the non-null error stored by this invocation's signaled FEL latch. The adapter certifies local completion when
\begin{equation}\label{eq:completion}
T_i = D_i \land \bigl((N_i \land V_i=G_v) \lor M_i\bigr).
\end{equation}
A failure before the first guard can satisfy $M_i$ with $V_i$ empty. A wait timeout followed by a later terminal signal does not establish $M_i$. The sets record identities alone. The checked linear stages determine visit order and multiplicity because each map forwards once on success and stops on error. Where the service owns a local response body, it also requires verified closure before releasing the lease as certified.

\texttt{InvocationWork} travels in the FEL context. Each invocation seals registration and awaits its own tickets under one cleanup budget, independently of other calls in its business session. A certified Guardian error releases the lease but can leave the session rejecting calls. To clear that state, \texttt{resetSession} holds selection, bars entry and checks for entering calls, active calls and pending guard work in the affected session. Other sessions can continue. Uncertainty closes the service and bars reset. Unreset metadata may persist until retirement.

Figure~\ref{fig:protocol} shows a successful switch. Strict publication drains old calls before reopening admission. When the retention budget permits, deferred publication admits calls to the successor while old leases remain. Both policies retire the predecessor under Eq.~\ref{eq:retire}. For native generation, the harness captures the outer version before dispatch. Release requires a provider terminal record, local body closure and certified inner completion. An incomplete exchange quarantines the outer gate.

\begin{figure}[t]
\centering
\begin{tikzpicture}[font=\fontsize{8}{9.5}\selectfont,>=Stealth,
  old/.style={black!75,line width=1.25pt},
  next/.style={blue!65!black,line width=1.25pt},
  event/.style={dashed,black!55,line width=.4pt},
  note/.style={inner sep=1pt}]
\node[anchor=west,font=\fontsize{8}{9.5}\selectfont\bfseries] at (0,0) {Strict publication};
\draw[old] (0,-.58) -- node[pos=.24,above=1pt,note] {Old calls} (4.05,-.58);
\draw[next,->] (4.05,-.58) -- node[midway,above=1pt,note] {New calls} (6.05,-.58);
\draw[event] (1.52,-.68) -- (1.52,-1.08);
\draw[event] (4.05,-.35) -- (4.05,-1.20);
\draw[<->,line width=.4pt] (1.52,-.96) -- node[below,note] {Admission paused} (4.05,-.96);
\node[anchor=west,note] at (4.12,-.96) {Select and resume};

\node[anchor=west,font=\fontsize{8}{9.5}\selectfont\bfseries] at (0,-1.59) {Deferred publication};
\draw[old] (0,-2.14) -- node[pos=.24,above=1pt,note] {Old calls} (4.05,-2.14);
\draw[next,->] (1.52,-2.62) -- node[midway,above=1pt,note] {New calls} (6.05,-2.62);
\draw[event] (1.52,-2.02) -- (1.52,-2.84);
\draw[event] (4.05,-2.14) -- (4.05,-2.36);
\node[anchor=west,note] at (4.12,-2.14) {Last old lease ends};
\node[anchor=west,note] at (1.57,-2.84) {Select new version};
\draw[->,line width=.4pt] (0,-3.31) -- node[above,pos=.93,note] {Time} (6.05,-3.31);
\end{tikzpicture}
\caption{Strict and deferred publication}\label{fig:protocol}
\end{figure}

\subsection{Composition rule and retirement}

Every admitted action atomically captures the selected version and increments its lease count $L_v$, then runs outside the selection lock. Let $\mathrm{Dirty}(v)$ denote a history of uncertified completion. The resource check $B_v$ requires zero active calls, signaled retained terminal latches, drained Guardians and verified closure of retained HTTP streams. Define retirement eligibility as
\begin{equation}\label{eq:retire}
\begin{aligned}
\mathrm{Eligible}(v)\equiv{}&v\ne\mathrm{Selected}\land L_v=0\\
 &{}\land\neg\mathrm{Dirty}(v)\land B_v.
\end{aligned}
\end{equation}
AiFlow uses Eq.~\ref{eq:retire} for automatic retirement. Explicit shutdown may dispose of quarantined resources after permanent closure, exit of all entering and executing calls, and direct verification of $B_v$. Failed checks retain handles for retry. Successful shutdown clears version-local state and drops version and coordinator references, retaining the last selection tag for diagnostics.

\noindent\textbf{Proposition (conditional local conformance).} Assume faithful graph inspection, owner-mediated capture, linear stages that forward once on success and stop on error, evidence for each local action that can use the version, and shared resources that outlive their users. Each admitted call then executes its captured action. A successful return covers all expected guards. A matched error can end the path earlier. At disposal, every covered local user has finished.

\textit{I1. Operator identity.} The factory retains private final references to the checked flow and its guards. A call invokes the action stored in the handle it captured. A subsequent publication changes which handle later calls receive.

\textit{I2. Invocation ownership.} Each guard enrolls in the capturing call's context before admission. Under Eq.~\ref{eq:completion}, a certified outcome accounts for the terminal event and release of every enrolled ticket. Work transferred to another executor remains covered through a join or a tracked release handle.

\textit{I3. Retirement.} A certified call releases its lease. An uncertain call closes admission before logical release. The strict gate stays closed, and the deferred gate also records uncertainty on the captured version. For a deselected version, zero leases and a clean history therefore account for every earlier capture. The resource check $B_v$ then completes Eq.~\ref{eq:retire}. Shutdown bars new captures permanently and verifies outstanding resources directly. Online Resource~1 traces the five transitions that preserve I1--I3.

Strict replacement pauses entry, drains leases and disposes of any older predecessor before selection. If resumption fails, admission stays paused. Successful recovery pruning bounds retained published versions at two. Preparation can transiently add a third handle. Deferred publication refuses further updates while two versions remain. Progress requires terminating work and successful cleanup.

Table~\ref{tab:assumptions} summarizes the owner protocol. The Guardian allowance $\sum_n(q_n+k_n+\rho_n)$ counts queue slots, workers and reorder slots. The ingress cap counts caller permits across versions. Payload bytes, remote buffers and detached work require separate accounting.

\begin{table}[t]
\caption{Version owner protocol}\label{tab:assumptions}
\centering\footnotesize\setlength{\tabcolsep}{3pt}
\begin{tabular}{@{}>{\raggedright\arraybackslash}p{.34\columnwidth}>{\raggedright\arraybackslash}p{.60\columnwidth}@{}}
\toprule
Event & Owner decision\\
\midrule
Publish & Claim a private checked pair. Strict pauses admission and drains old calls before switching. Deferred switches when retention permits and keeps versions still in use.\\
Admit request & Capture the selected version and increment its lease atomically.\\
Normal return & Require delivery to the consumer and enrollment of every expected guard.\\
Failed wait & Require the wait exception's direct cause to be the same error object recorded by this call's signaled latch. An error may precede the first guard.\\
Certify either outcome & Once enrollment is sealed, verify release of all issued tickets and closure of any owned response body before releasing the lease.\\
Uncertified outcome & Close admission and quarantine the version before logical lease release.\\
Automatic retirement & Require deselection, zero leases, clean status and verified local release.\\
Explicit shutdown & Bar entry and require zero entering or executing calls. Check terminals, Guardians and retained bodies before disposal.\\
\botrule
\end{tabular}
\end{table}

\subsection{Reference implementations and bounded model}

Two reference implementations share AiFlow's version coordinators. One wires preparation and execution separately. The other binds the checked action in an immutable handle. Both track local work with per-call futures and a bounded \texttt{ThreadPoolExecutor} \cite{refJdkFuture,refJdkPool}. After a failed wait, they check the signaled latch and completed futures. AiFlow's completion rule in Eq.~\ref{eq:completion} additionally requires the direct cause of the wait exception to be the error object recorded for that invocation.

A third reference coordinates prechecked candidates itself. A management lock serializes publication and disposal, and a monitor protects capture and leases. Strict replacement holds admission closed through drain, disposal and selection. The reference and AiFlow share the FEL workload, endpoint and schedules (Sect.~\ref{sec:reference}).

The bounded model covers capture, enrollment, result/release/terminal ordering, certification, closure and shutdown with three versions and two one-shot callers. Strict selection and resumption are separate transitions, including interruption and recovery pruning. The model explores 10,996 reachable states under strict publication and 30,664 under deferred publication. It checks a bound of two retained published versions, assumes successful disposal checks and counts a prepared candidate only after publication. Wrong-object execution, terminal-only and release-only certification yield counterexamples in both modes. Omitting strict recovery pruning retains three published versions after interrupted resumption.

\section{Evaluation method}\label{sec:method}

We evaluate fault containment (RQ1), lifecycle decisions and completion cost (RQ2), publication timing (RQ3), and ownership across a model stream or asynchronous audit (RQ4).

We replay the upstream parser change with its original filtering helper. The test harness supplies the FEL port, loopback HTTP endpoint, fault injections and event schedules. FIT and app-platform are pinned at \texttt{e2f285d1} and \texttt{33d5283}. Timing trials pair strategies within JVM runs \cite{ref40}. These runs use fixed 1 GB heaps on an Apple M1 Pro with macOS 26.6.2 and OpenJDK 21.0.10. They run sequentially and separately from model generation. \textit{Data and code availability.} Online Resource~1 contains the code, inputs, raw traces, expected outputs and analysis scripts.

\noindent\textbf{Conformance and ownership controls.} Six fixtures cover plain/reasoning-prefixed arrays, multiline/repeated blocks, Unicode and whitespace. Four ablation cases---valid replacement, ordinary-map bypass, held old audit, and both faults---repeat three times in four construction/completion modes: 48 observations. All initially install valid guards. Bypass preserves business output with a record from the previous flow. Held work ignores interruption. Fault calls/cleanup use 150 ms budgets; normal calls use three seconds. The ablation varies both checks through the lower-level interface. With completion checking off, the driver waits for the terminal result and guard drain. With it on, the original trials check the terminal latch and release of registered work. We repeated the same configurations with the cause-matching adapter of Eq.~\ref{eq:completion}. The public facade applies both checks.

One legal and four ill-typed clients test the closed API. Seven runtime definitions repeat three times: two-guard replacement, cross-owner reuse, client disposal, shutdown before entry, active retirement and uncertainty in both modes. Six callback controls test preparation/action mismatch and mutable flows; ten regressions cover interruption, projection, terminal retention, shutdown and transport faults.

\noindent\textbf{Reference implementations and batch timing.} The independent static checker covers topology, identity, witnesses, installation and policy. Four sizes and 14 variants give 56 configuration cases (16 legal and 40 mismatched), each repeated three times, yielding 168 observations. Eight completion scenarios cover held work, interruption, a certified error and concurrent calls in one session. Three repetitions per implementation give 48 observations. The guard uses eight slots and one worker. Uncertain call and cleanup budgets are 150 ms. Other call budgets are three seconds, or five for shared-session controls.

The batch timings, fixed-window and eight-JVM timing series, and Ollama calls use an adapter that checks the terminal latch and work release after a failed wait. The adapter implementing Eq.~\ref{eq:completion} also matches the exception cause. The ordered failure controls and the two-interval experiment use this rule. Both adapters use the same normal completion path. Each lifecycle comparison repeats four cases three times per policy. The cases cover valid replacement, held audit, uncertainty and certified audit error. Six probes substitute another action in the separate-callback reference. Eight additional cases on the separately implemented coordinator, each run three times, cover ownership, interrupted or expired replacement, retention, uncertainty, close, cleanup retry and label reuse.

Eight fresh JVMs each measure 30 pairs of eight-request batches after ten warm-up pairs, balancing arm order by seed. Both arms use the same workflow instance, Guardian, gate, HTTP client/body reader and executors. The terminal-only arm measures successful batches that reach quiescence. Fault handling requires evidence of work release. Timing spans dispatch through collection of all eight returns. Twenty samples per phase and JVM average ten revalidations of the prebuilt graph or 100 idle release checks over eight completed handles.

\noindent\textbf{Fixed-window update sensitivity.} Three fresh JVMs run both strategies at three tail-service times (60/240/800 ms) and four update intervals (none/100/400/1,200 ms). Condition order is seeded and both arm orders appear per condition. Each cohort offers 192 requests at 80/s over 2.4 s. Every fourth request receives the tail delay in the HTTP service. Others receive seeded 0/2/5/20 ms delays. The 72 cohorts comprise 13,824 offers and 522 scheduled updates.

Both strategies use the facade, an ingress cap of 18, an eight-thread guard executor and a fresh eight-worker HTTP endpoint per cohort. Each version has one guard permit and an eight- or four-slot queue, alternating across epochs. Strict serves one version at a time. Deferred can use both versions' resources. Each interval yields 23/5/1 due updates per cohort. The serialized controller expires updates more than 300 ms late at start and never retries. Strict drain has a five-second budget and deferred cleanup is polled every 5 ms. One warm-up per cohort is excluded.

Three fresh JVMs vary update offsets from aligned tail arrivals by 0, 12.5, 25 and 37.5 ms at the 800 ms tail/400 ms interval, using matched inputs and both arm orders: 24 cohorts, 4,608 offers and 120 updates.

Two 8 s series use the 800 ms tail and 400 ms updates. Each 50 ms frame has four offers. Periodic arrivals are evenly spaced, with the tail first. Bursts put all four in the first 12.5 ms with a seeded tail position. Eight fresh JVMs per series cross profile and strategy order. Each occupies each position twice. Inputs and due times match by JVM index and profile.

Request deadlines of 100, 500 and 1,000 ms run from scheduled arrival. The caller samples time immediately after the service returns, before result validation, accounting and permit release. Refusals, errors and uncertainty miss the deadline. Update readiness is strict resumption or deferred selection. A scheduled update meets the 300 ms target if it is ready within 300 ms of its due time. Let $v$ be the captured version's schedule index and $j(t)$ the latest index due at time $t$, including refused updates. When $v<j(t)$, epoch age is $t-d_{v+1}$, measured from the first missed due time. It is zero otherwise. Entry p95 uses capture time among completed requests. Joint metrics use epoch age, exact epoch or queue-capacity agreement at return and count timely requests over all offers.

\noindent\textbf{Matched resource budgets.} Both arms have one guard executor thread, eight HTTP workers and 18 ingress permits. Strict gives its active filter eight queue slots. Deferred gives each active filter four, up to eight across two versions. Each filter has one dispatch permit, giving strict one and deferred up to two during overlap. Both still execute guard work on one thread. The short-window study at 240/100 and 800/400 ms adds a deferred arm with eight slots per version and two threads. Three JVMs balance the three arm positions. In the matched 8 s series, each arm preserves its filter and queue configuration across updates. The latest-version target measures adoption of the latest scheduled version due at return. The paired arms receive the same arrivals. A third 8 s series lowers the matched-budget arrival rate to 20/s: 160 offers and 19 due updates per cohort. Periodic offers are 50 ms apart; bursts place four offers in the first 50 ms of each 200 ms frame. One in four still receives 800 ms delay. Eight fresh JVMs use the same crossed orders, yielding 32 cohorts and 16 pairs. These tails require four worker-seconds/s at the eight-worker endpoint and impose a 75\% ceiling on 500 ms attainment. Each collection session pairs the two strategies at one arrival rate.

\noindent\textbf{Trace retention under failure.} The runner retains cohort outcomes and exception diagnostics. Three untimed cohorts hold an old audit. One strict cohort and one deferred cohort force an uncertain call. A second strict cohort sets the drain budget to 50 ms.

\noindent\textbf{Behavioral update replay.} For maintenance change \texttt{f244ef0}, the old version forwards raw answers; its replacement applies the unchanged filter. The endpoint reports HTTP 400 where upstream returns an empty list. Three reasoning-prefixed fixtures hold an old call during replacement in each mode; a plain array first checks the old path. We check errors, version tags, new output and reclamation.

\noindent\textbf{Native lifecycle boundaries.} Three Ollama \texttt{qwen2.5:3b} runs use server 0.22.0, temperature 0, seed 42, 4,096-token context and a 2,048-token cap. Array schemas constrain output \cite{refStructured}; expected arrays are specified independently. Each run tests a 64-string answer, a read aborted after eight nonempty records, a held audit and a fresh-version control. Short answers contain three strings; one unmeasured token loads the model.

The outer lease spans model dispatch, local body closure and the inner call. In the long case, the reader pauses after the first packet. Strict replacement is refused until the reader resumes and the call completes. In the held-audit case, work outlasts the 150 ms call and cleanup budgets and blocks disposal. The reader requires Ollama's terminal record \cite{ref60}. A separate 180 s deadline closes the raw stream even on an incomplete record. Abort, read or close failure and incomplete exchanges yield uncertainty. Online Resource~1 retains packets, answers and model snapshots.

In the audit controls a stage submits work to another executor. With an explicit future join, the stage remains active while a latch holds the worker through replacement or a 500 ms caller timeout. The detached case lets the stage return with the worker still held. Three fixtures per case and mode test the ordering of worker release and version disposal.

The author used OpenAI ChatGPT and Codex for code, script and manuscript revisions and for consistency checks. The author takes responsibility for the code, data and conclusions.

\section{Results}\label{sec:results}

\subsection{RQ1: Construction and completion checks}\label{sec:joint}

Both checkers accept all 48 legal trials and reject all 120 mismatch trials. An identity-only ablation accepts 72 of the mismatch trials that the full checkers reject. The additional checks cover topology, policy and installation.

Table~\ref{tab:joint} counts two violations under strict publication. In each row, C and T indicate whether construction and completion checks are enabled. Each of the three fault scenarios is repeated three times. A successor may be selected while the old audit is still active, or a request may reach the audit after successful HTTP processing with zero guard admissions. An observation is contained if neither event occurs. All four modes accept the three valid controls.

\begin{table}[t]
\caption{Fault containment in strict updates}\label{tab:joint}
\centering\footnotesize
\setlength{\tabcolsep}{3pt}
\begin{tabular}{@{}ccrrrr@{}}
\toprule
C & T & Valid & Overlap & Unguarded & Faults contained\\
\midrule
0 & 0 & 3/3 & 6 & 6 & 0/9\\
0 & 1 & 3/3 & 0 & 3 & 6/9\\
1 & 0 & 3/3 & 3 & 0 & 6/9\\
1 & 1 & 3/3 & 0 & 0 & 9/9\\
\botrule
\end{tabular}
\par\smallskip
\begin{tabular}{@{}>{\raggedright\arraybackslash}p{.41\columnwidth}@{\hspace{.03\columnwidth}}>{\raggedright\arraybackslash}p{.54\columnwidth}@{}}
\toprule
Mechanism & What the controls show\\
\midrule
Graph check and lease, separate action & A substitute reaches HTTP and audit with zero guard entries.\\
Bound action, wait for caller & The old audit can still run when the caller returns.\\
Bound action, tracked work and retirement check & The separately coordinated reference matches AiFlow in 24 tested lifecycle pairs (Sect.~\ref{sec:reference}).\\
\botrule
\end{tabular}
\end{table}

With only the completion check enabled, all three filter substitutions passed the HTTP result check and reached the audit. The missing guard enrollment then prevented a successful return. With only the construction check enabled, all three held-audit trials selected a successor while the old audit was active. In the combined fault, either check blocked successor selection. Across the three injected scenarios, neither monitored violation occurred in six of nine trials with either check alone and in all nine with both. The cause-matching replay reproduced all 48 ablation records, including the valid controls.

A test hook exposed the flow at four graph sizes. In five repetitions per size, all 20 ordinary-map edits preserved the named topology and bypassed the guard. All 20 regenerations through the builder retained guarded execution. Distinct output tags identify the operator that ran in these controlled edits.

All four ill-typed clients fail compilation, the legal client compiles, and all twenty-one runtime controls pass. Shutdown refuses disposal before terminal registration. Six callback misuse controls deliver expected output without guard admission; the closed interface prevents these substitutions.

A second binding of the same guard and a duplicate installation were both rejected. In a separate ledger test, two enrollments of one guard required two ticket releases. Five interrupted resumptions leave two owned versions, and the successful retry reclaims both predecessors. Projection controls show that distinct tags distinguish versions reusing a label. A valid retry in the failed session is rejected until explicit reset, then succeeds. The terminal set stays empty after each of 4,096 certified calls. In each publication mode, a successful call on the two-guard chain entered both guards and the sink. Stops before and inside the first guard recorded zero and one guard admission, respectively. Neither error reached the audit or sink, and a fresh session completed after each error.

All six behavioral update replays accept the plain array on the old version. For the reasoning-prefixed input, the held old call reaches HTTP 400 and produces no delivery. The new version returns the expected array and its own tag. Strict first refuses replacement, then succeeds after the old error. Deferred publishes and delivers the new result while the old call is held, then retires it after the error. Both errors have certified termination. The version tags identify which filter handled each input, including the old call that failed after deferred publication.

\subsection{RQ2: References and batch completion times}\label{sec:reference}

Table~\ref{tab:reference} compares two ways of tracking invocation-owned work on the same workflow, with 24 observations per implementation. Successes include normal and shared-session calls. Errors count injected terminal errors. In the shared-session control, one call returns while the other remains active and later finishes normally.

\begin{table}[t]
\caption{Ticket and future outcomes}\label{tab:reference}
\centering\footnotesize
\setlength{\tabcolsep}{3pt}
\begin{tabular}{@{}lrrrr@{}}
\toprule
Implementation & Success & Errors & Uncertain & Unsafe\\
\midrule
AiFlow tickets & 6 & 3 & 15 & 0\\
JDK owned futures & 6 & 3 & 15 & 0\\
\botrule
\end{tabular}
\end{table}

Table~\ref{tab:cost} isolates completion certification in 240 pairs (3,840 measured requests). Each row gives median times in milliseconds for eight-request batches over 30 pairs in one JVM. We subtract the terminal-only time from the certified time within each pair and linearly interpolate the quartiles of these differences. Per-JVM paired medians range from $-0.437$ to $+0.224$ ms, with all eight interquartile intervals spanning zero. Per-JVM median revalidation times range from 77.108 to 113.158 $\mu$s; idle release checks range from 0.433 to 0.575 $\mu$s. Full lifecycle cost also includes construction, installation, active waiting and disposal.

With AiFlow's coordinators held fixed, both the separate-callback reference and the immutable-handle reference matched the facade in 24 lifecycle pairs each. In all six manual substitution probes, the separate callbacks allowed one graph to pass its check and another to execute. The private handle bound those operations to one workflow.

The reference with its own coordinator matched AiFlow in all 24 lifecycle pairs on outcome, HTTP call count and old-version reclamation. The eight boundary cases each produced the expected outcome in all three repetitions. The comparison used the latch and release checks described in Sect.~\ref{sec:method}. Repeating the four scenarios under both policies with exception cause matching yielded the same 48 observations.

\begin{table}[t]
\caption{Batch completion times}\label{tab:cost}
\centering\footnotesize\setlength{\tabcolsep}{2.5pt}
\begin{tabular}{@{}rrrrr@{}}
\toprule
JVM & Terminal & Certified & Median diff. & Q1 to Q3\\
\midrule
1 & 99.364 & 99.050 & 0.015 & $-0.919$ to $+0.457$\\
2 & 99.149 & 99.199 & 0.051 & $-1.186$ to $+1.194$\\
3 & 98.965 & 99.330 & 0.224 & $-0.352$ to $+1.164$\\
4 & 99.248 & 99.329 & -0.046 & $-1.261$ to $+1.191$\\
5 & 99.480 & 99.063 & -0.351 & $-1.844$ to $+0.902$\\
6 & 98.752 & 98.388 & -0.437 & $-2.435$ to $+1.323$\\
7 & 99.058 & 98.424 & -0.420 & $-1.775$ to $+1.079$\\
8 & 99.135 & 99.040 & -0.350 & $-1.020$ to $+1.616$\\
\botrule
\end{tabular}
\end{table}

\subsection{RQ3: Request and update timing}\label{sec:rolling}

The fixed-window trials record 13,824 offers and 522 scheduled updates. There are 9,929 completed requests, including 791 that exceed 500 ms. The series records 0 business errors and 0 uncertain outcomes. Other unsuccessful offers are gate or request-budget refusals. Every cohort respects the cap of eighteen caller permits and the strict/deferred limits of one/two active versions.

In Table~\ref{tab:rolling}, each setting pools 576 offers per strategy across three JVMs. Throughout the timing results, paired entries give strict publication first. Late returns exceed 500 ms. Joint counts timely returns with the latest epoch due at return. The entry age percentile includes completed calls with zero age. At update intervals of 100, 400 and 1,200 ms, respectively, 69, 15 and 3 updates are due per strategy. Ready counts updates ready within 300 ms. Deferred has more 500 ms and latest-epoch joint successes in 8 of 9 update conditions. Strict leads on that joint target in the remaining condition (800 ms tail, 400 ms updates). At a 60 ms tail and 100 ms updates, strict and deferred meet 500 ms for 235 and 576 of 576 offers and publish 69 and 69 of 69 updates. Paired exact-epoch counts reverse across JVMs at 800/1200 ms (+18, -5, +18). Differences are deferred minus strict per 192 offers. Both strategies publish 3 of 3 updates in this condition. Strict meets 0 of 3 update deadlines and deferred meets 3 of 3.

\begin{table*}[t]
\caption{Effects of tail latency and update interval}\label{tab:rolling}
\centering\footnotesize\setlength{\tabcolsep}{4pt}
\begin{tabular}{@{}rrrrrrrrr@{}}
\toprule
Tail (ms) & Interval (ms) & Completed & 500 ms (\%) & Late & Published & Ready & Stale p95 (ms) & Joint\\
\midrule
60 & None & 576/576 & 100.0/100.0 & 0/0 & 0/0 & 0/0 & 0.0/0.0 & 576/576\\
60 & 100 & 235/576 & 40.8/100.0 & 0/0 & 69/69 & 69/69 & 2.7/2.5 & 108/384\\
60 & 400 & 501/576 & 87.0/100.0 & 0/0 & 15/15 & 15/15 & 0.0/0.0 & 448/524\\
60 & 1,200 & 561/576 & 97.4/100.0 & 0/0 & 3/3 & 3/3 & 0.0/0.0 & 550/565\\
240 & None & 576/576 & 100.0/100.0 & 0/0 & 0/0 & 0/0 & 0.0/0.0 & 576/576\\
240 & 100 & 113/576 & 19.6/100.0 & 0/0 & 67/25 & 62/25 & 3.2/189.9 & 62/148\\
240 & 400 & 279/576 & 48.4/100.0 & 0/0 & 15/15 & 15/15 & 0.2/0.0 & 208/478\\
240 & 1,200 & 518/576 & 89.9/100.0 & 0/0 & 3/3 & 3/3 & 0.0/0.0 & 496/555\\
800 & None & 293/290 & 28.1/29.0 & 131/123 & 0/0 & 0/0 & 0.0/0.0 & 162/167\\
800 & 100 & 70/304 & 8.7/31.4 & 20/123 & 27/9 & 18/9 & 2.7/852.5 & 39/45\\
800 & 400 & 185/301 & 22.9/30.6 & 53/125 & 9/6 & 3/6 & 0.0/1140.1 & 116/76\\
800 & 1,200 & 220/299 & 23.6/29.0 & 84/132 & 3/3 & 0/3 & 0.0/0.0 & 132/163\\
\botrule
\end{tabular}
\end{table*}

At a 240 ms tail and 100 ms updates, strict and deferred deliver 113 and 576 responses within 500 ms and publish 67 and 25 of 69 updates. Entry stale-age p95 is 3.2 and 189.9 ms among completed requests. The latest-epoch joint counts are 62 and 148. The queue-capacity joint counts are 69 and 340.

At an 800 ms tail with no updates, across both strategies, 122/583 completed requests meet 100 ms, 329 meet 500 ms and 549 meet 1,000 ms.

At an 800 ms tail and 400 ms updates, strict and deferred meet 500 ms for 132 and 176 of 576 offers. Requiring the latest due epoch gives 116 and 76 of 576. Allowing 100 ms of age gives 120 and 83 timely responses. Allowing 400 ms of age gives 122 and 130, with per-JVM differences of +3, $-4$, and +9 (deferred minus strict). Queue-capacity matching gives 126 and 100 timely responses.

In that condition, strict records 391 gate and 0 budget refusals. Deferred records 0 gate and 275 budget refusals. Peak active-version counts across JVMs are 1 for strict and 2 for deferred.

Table~\ref{tab:offset} varies the update offset at an 800 ms tail and a 400 ms update interval. At each offset, each strategy receives 576 offers and 15 scheduled updates. The epoch and capacity columns add their respective version requirements to the 500 ms deadline. The final two columns count updates ready within 300 ms and all published updates. Deferred leads on the 500 ms deadline alone at 4 of 4 pooled offsets and in 10 of 12 matched JVM/offset pairs. Strict leads on the latest-epoch joint target at 4 of 4 offsets and in 12 of 12 pairs. Strict leads on the queue-capacity target at 4 of 4 pooled offsets. With a 400 ms age allowance the strict/deferred joint counts are 120/127, 121/117, 124/112, 129/123 in offset order. Strict meets fewer update deadlines at every offset, with 3 of 15 on time versus 6 of 15 for deferred. Online Resource~1 gives per-JVM counts.

\begin{table}[t]
\caption{Deadline attainment by update offset}\label{tab:offset}
\centering\footnotesize\setlength{\tabcolsep}{3pt}
\begin{tabular}{@{}rrrrrr@{}}
\toprule
Offset (ms) & 500 ms & + Epoch & + Capacity & 300 ms & Publ.\\
\midrule
0 & 127/179 & 117/82 & 124/106 & 3/6 & 9/6\\
12.5 & 130/170 & 116/68 & 124/82 & 3/6 & 9/6\\
25 & 134/171 & 119/73 & 129/97 & 3/6 & 9/6\\
37.5 & 140/162 & 122/76 & 132/87 & 3/6 & 9/6\\
\botrule
\end{tabular}
\end{table}

The matched-budget study records 3,456 offers and 252 updates (Table~\ref{tab:budget}). Each arm receives 576 offers per setting. The larger deferred arm adds queue slots and a second guard executor thread. The epoch column counts timely responses on the latest version due at return. At 240/100 ms, strict publication meets more update deadlines, while deferred delivers more timely responses using the latest due version. At 800/400 ms, strict delivers more of those responses, while deferred meets more update deadlines. These directions hold in all three JVMs for the two matched arms. At the longer setting, strict publishes more updates in total, including updates that become ready after the deadline.

\begin{table}[t]
\caption{Deadline attainment by guard budget}\label{tab:budget}
\centering\footnotesize\setlength{\tabcolsep}{3pt}
\begin{tabular}{@{}lrrrr@{}}
\toprule
Mode & \shortstack{500 ms\\requests} & \shortstack{500 ms\\+ epoch} & \shortstack{300 ms\\updates} & \shortstack{All\\published}\\
\midrule
\multicolumn{5}{@{}l}{240 ms tail, 100 ms interval (69 updates per arm)}\\
Strict & 133 & 73 & 62 & 66\\
Deferred & 576 & 150 & 25 & 25\\
Deferred, larger & 576 & 142 & 24 & 24\\
\midrule
\multicolumn{5}{@{}l}{800 ms tail, 400 ms interval (15 updates per arm)}\\
Strict & 127 & 116 & 3 & 9\\
Deferred & 189 & 84 & 6 & 6\\
Deferred, larger & 174 & 77 & 6 & 6\\
\botrule
\end{tabular}
\end{table}

Each 80/s 8 s series records 20,480 offers and 608 updates (Table~\ref{tab:long}). For each arrival profile and strategy, 152 updates are due and the number of offers is 5,120 at 80/s or 1,280 at 20/s. Request percentages use all offers, including refusals at the strict gate or the deferred request budget. At 80 offers/s, the injected tails alone offer 16 worker-seconds per second to an eight-worker HTTP endpoint. In each series, deferred leads on the 500 ms deadline in 16 of 16 pairs, and strict leads in 16 of 16 when the latest scheduled version due at return is also required. For example, a timely call captured on $v_1$ remains safe on $v_1$ when $v_2$ falls due during execution, but misses the joint target. In the matched series, strict leads in 15 of 16 pairs at 400 ms of allowed age, and deferred leads in 14 of 16 at 800 ms. Online Resource~1 reports all paired differences.

\begin{table*}[t]
\caption{Deadline attainment by load and arrival pattern}\label{tab:long}
\centering\footnotesize
\setlength{\tabcolsep}{3pt}
\begin{tabular}{@{}lrrrrrr@{}}
\toprule
 & \multicolumn{2}{c}{Per-version, 80/s} & \multicolumn{2}{c}{Matched, 80/s} & \multicolumn{2}{c}{Matched, 20/s}\\
\cmidrule(lr){2-3}\cmidrule(lr){4-5}\cmidrule(l){6-7}
Outcome & Periodic & Burst & Periodic & Burst & Periodic & Burst\\
\midrule
Requests $\leq$ 500 ms & 1164/1584 & 1218/1627 & 1236/1622 & 1215/1557 & 336/960 & 319/960\\
\quad as \% of offers & 22.7/30.9 & 23.8/31.8 & 24.1/31.7 & 23.7/30.4 & 26.2/75.0 & 24.9/75.0\\
\quad + scheduled-version target & 1058/369 & 1102/431 & 1174/438 & 1125/429 & 336/408 & 310/494\\
\quad + age $\leq$ 400 ms & 1126/883 & 1195/953 & 1209/995 & 1170/930 & 336/756 & 319/928\\
\quad + age $\leq$ 800 ms & 1163/1230 & 1218/1343 & 1236/1370 & 1215/1308 & 336/960 & 319/960\\
Admission refusals (\%) & 69.4/50.4 & 67.0/51.4 & 67.5/49.7 & 67.3/52.1 & 65.4/0.0 & 66.4/0.0\\
\midrule
On-time updates & 44/49 & 48/52 & 46/51 & 46/54 & 53/60 & 92/76\\
Published updates & 94/49 & 102/52 & 96/51 & 98/54 & 109/60 & 149/76\\
\botrule
\end{tabular}
\par\smallskip
\begin{minipage}{0.94\textwidth}
Under the per-version budget, deferred publication rejects 203 of 304 scheduled updates at the two-version retention limit. Under strict publication, 108 of 304 expire because the controller reaches them more than 300 ms after they were due. Under matched budgets at 80/s these counts are 199 of 304 and 110 of 304, respectively; at 20/s they are 168 of 304 and 46 of 304.

\smallskip
{\raggedright
\textit{Matched-budget differences at 80/s, deferred minus strict per 640 offers, in JVM order.}\par
Periodic deadline $(29,73,34,39,44,41,60,66)$ and joint $(-109,-69,-98,-98,-111,-87,-98,-66)$.\par
Burst deadline $(44,35,60,41,43,53,45,21)$ and joint $(-75,-58,-88,-112,-95,-62,-99,-107)$.\par}
\end{minipage}
\end{table*}

At 80/s, strict publishes more updates in both budgets and profiles, and deferred makes more ready within 300 ms. At 20/s the update interval remains 400 ms. This series records 5,120 offers and 608 updates, with no business errors or uncertainty. Deferred completes all 2,560 offers and reaches the 75\% deadline ceiling in each profile without admission refusal. Strict attains 26.2\%/24.9\% and refuses 65.4\%/66.4\% under periodic/burst arrivals because publication pauses admission. Deferred serves more offers meeting both targets in all 16 pairs, with differences of +3 to +18 per periodic cohort and +17 to +28 per burst cohort. Across both profiles, 646/2,560 strict offers (25.2\%) and 902/2,560 deferred offers (35.2\%) meet both targets. Among returns within 500 ms, 646/655 (98.6\%) under strict and 902/1,920 (47.0\%) under deferred carry the latest scheduled version. Each policy includes its own admission rule and allocation of queues and dispatch permits.

To isolate update frequency at 20 offers/s, we replayed the same inputs and arrival times with update intervals of 400 and 1,200 ms using the adapter defined by Eq.~\ref{eq:completion}. Two fresh JVMs reverse interval order, profile order and strategy order. Each interval has 640 offers per strategy across the two profiles. Lengthening the interval raises strict deadline attainment from 156 to 205 and reduces strict refusals from 431 to 366. Deferred attains 480 at both intervals. Timely returns on the latest scheduled version rise from 156 to 204 for strict and from 225 to 479 for deferred.

Each of the three fault cohorts retains 192 request outcomes and all 23 scheduled updates. One uncertain call in each policy's cohort closed admission. Strict recorded 23 of 23 update attempts as closed and deferred 19 of 23. In the separate strict drain-timeout cohort, four updates published after the held call released its work.

\subsection{RQ4: Native generation and audit handoff}\label{sec:native}

In Table~\ref{tab:native}, the harness supplies the outer gate around Ollama and the inner FEL service. The token count comes from the terminal provider record, which an aborted read never reaches. Returns count successful callers and Closed counts quarantined outer gates. Long answers stop naturally after 450 generated tokens, below the 2,048-token cap. In the strict long-answer trial, the old lease blocks the switch until the reader sees the terminal record, closes the body and completes the business call. Fresh calls execute the checked replacement. All twelve calls used the same model digest and server version.

\begin{table}[t]
\caption{Lifecycle outcomes for Ollama calls}\label{tab:native}
\centering\footnotesize\setlength{\tabcolsep}{3pt}
\begin{tabular}{@{}lrrrrr@{}}
\toprule
Case & Calls & Tokens & HTTP & Returns & Closed\\
\midrule
Long answer & 3 & 450 & 3 & 3 & 0\\
Read abort & 3 & -- & 0 & 0 & 3\\
Held audit & 3 & 8 & 3 & 0 & 3\\
Fresh version & 3 & 8 & 3 & 3 & 0\\
\botrule
\end{tabular}
\end{table}

The internal FEL consumer received 0/3 results after read aborts. Each client closed the body after eight records and quarantined the instance. The client trace ends at local body closure. The server controls any remaining generation. With a held audit, all three callers had already received an uncertain outcome when the harness released the audit after a refused disposal attempt. The internal FEL consumer then received 3/3 results. Admission stayed closed through cleanup.

Nine calls reached HTTP and six returned successfully to the caller. Schema-constrained copying exercises filter passthrough. Fixed reasoning-prefixed fixtures exercise removal.

Four model-reader controls exercise truncation, end of stream without a terminal record, a close error and a stalled HTTP body. The stalled peer sends one complete and one partial record. The 500 ms body watchdog closes the local stream and the outer gate quarantines before peer release. Each fault prevents business delivery.

The facade sets separate 5 s request and body budgets. Two loopback trials use a 750 ms body budget. When the body timeout occurs first, the invocation reports a certified error and the next call proceeds. An earlier 250 ms caller timeout quarantines the instance. Shutdown then verifies body closure. A failed close retains its handle for retry. In both publication modes, a matching FEL error with released work left admission open. A timeout followed by normal delivery, an interrupt followed by error delivery, and a matching error with a held ticket each closed admission. These eight observations cover the four failure cases under both publication policies. Online Resource~1 identifies the adapter and dependencies for each series.

In the 18 asynchronous audit controls, the six joined calls retain their versions while the worker is held. Strict refuses three replacements, while deferred publishes three and serves the new version. The six joined timeouts close admission and block shutdown until worker release. In all six detached controls, disposal preceded worker release. The detached worker waits on a latch, then signals its exit. Joining makes that exit observable through FEL terminal delivery.

\section{Discussion and applicability}\label{sec:discussion}

Moving from recommendation filtering to document projection changes the private factory that declares and binds the guarded operator. Both workflows reuse the invocation context, completion adapter and version owner. The deadline measurements characterize the recommendation workflow, whose HTTP call and audit both use the captured version.

The bounded model treats capture and selection as atomic steps for three versions and two callers. In the single-JVM implementation, locks serialize the corresponding operations. Extending the contract to a branching workflow would require identifying the branches that can still use a version after an outcome is delivered. A stateful update would also require a rule for transferring state.

When a call ends with an uncertain outcome, the instance closes admission and retains the handles needed to check outstanding work at shutdown. Service availability during this interval depends on an external manager that can restart the instance or route traffic elsewhere.

Policy choice depends on the service objective and the offered load (Table~\ref{tab:long}). In the eight-second windows at 80 offers/s under matched budgets, 22.5\% of strict offers and 8.5\% of deferred offers met both targets across the two profiles. At 20/s, strict still had a larger share of current versions among timely returns, but deferred admitted enough additional calls to deliver more timely, current-version responses overall. Because the matched-budget trials keep the filter fixed across updates, agreement with the schedule measures version adoption. The parser experiment shows how a change in filter behavior affects the output of an individual call.

\section{Conclusion}\label{sec:conclusion}

AiFlow connects the identity of a checked workflow operator to the lifetime of the version that contains it. Binding each admitted call to its checked action preserves the execution path. Certifying the call's terminal outcome and the release of its registered work gives the owner evidence to retire a deselected version once its local resources are clear. The separately implemented coordinator reached the same lifecycle decisions in the tested scenarios. Strict and deferred publication can thus use a common retirement rule while making different tradeoffs between timely responses and adoption of scheduled versions.

\section*{Statements and Declarations}

\noindent\textbf{Funding.} The author received no specific funding for this work.

\noindent\textbf{Competing interests.} The author declares no known competing interests.

\noindent\textbf{Author contributions.} Qunhui Zhang conceived the study, led the method and software development, conducted the main experiments and data analysis, and wrote the original draft.

\end{document}